\documentclass[11pt]{article}
\usepackage[utf8]{inputenc}
\usepackage{fullpage}
\usepackage{amsthm}
\usepackage{amsmath}
\usepackage{amssymb}
\usepackage{mathtools}
\usepackage[colorlinks,citecolor=blue,bookmarks=true,linktocpage]{hyperref}
\usepackage{aliascnt}
\usepackage[numbered]{bookmark}
\usepackage[capitalise]{cleveref}
\usepackage[backend=biber,url=false,arxiv=abs,maxbibnames=100,isbn=false,sortcites=true]{biblatex}
\usepackage{cleveref}
\usepackage{xcolor}

\newtheorem{lemma}{Lemma}
\newtheorem{theorem}{Theorem}
\newtheorem{corollary}{Corollary}
\theoremstyle{remark}
\newtheorem*{remark}{Remark}

\newcommand{\head}{{\rm head}}
\newcommand{\tail}{{\rm tail}}

\title{Reconfiguration of Time-Respecting Arborescences\thanks{This work was partially supported by 
JSPS KAKENHI Grant Numbers 
JP18H04091, %% Ito, Suzuki (Uehara KIBAN A)
JP19K11814, %% Ito (KIBAN C)
JP20H05793, %% CORE A01
JP20H05794, %% CORE B01
JP20H05795, %% CORE C01
JP20K11666, %% Suzuki (KIBAN C)
JP20K11692, %% Yusuke Kobayashi KIBAN C
JP20K19742, %% Yasuaki Kobayashi Wakate
JP20K23323, %% Yuni Iwamasa Kensuta
JP22K17854, %% Yuni Iwamasa Wakate
JP22K13956. %% Shunichi Maezawa Wakate
}}

\author{Takehiro Ito\thanks{
Graduate School of Information Sciences, 
Tohoku University, Sendai, Japan.
\texttt{\{takehiro, akira\}@tohoku.ac.jp}}\and
Yuni Iwamasa\thanks{
Graduate School of Informatics, 
Kyoto University, Kyoto, Japan.
\texttt{iwamasa@i.kyoto-u.ac.jp} }\and
Naoyuki Kamiyama\thanks{
Institute of Mathematics for Industry, 
Kyushu University, Fukuoka, Japan.
\texttt{kamiyama@imi.kyushu-u.ac.jp}}\and 
Yasuaki Kobayashi\thanks{
Graduate School of Information Science and Technology, 	Hokkaido University, Sapporo, Japan.
\texttt{koba@ist.hokudai.ac.jp}}\and
Yusuke Kobayashi\thanks{
Research Institute for Mathematical Sciences, 
Kyoto University, Kyoto, Japan.
\texttt{yusuke@kurims.kyoto-u.ac.jp}}\and
Shun-ichi Maezawa\thanks{
Department of Mathematics, 
Tokyo University of Science, Tokyo, Japan.
\texttt{maezawa.mw@gmail.com}}\and  
Akira Suzuki\footnotemark[2]
}

\begin{document}
\maketitle 

\begin{abstract}
An arborescence, which is a directed analogue of a spanning tree in 
an undirected graph, is one of the most fundamental combinatorial objects 
in a digraph. In this paper, we 
study arborescences in digraphs from the viewpoint of combinatorial 
reconfiguration, which is the field where 
we study reachability between two 
configurations of some combinatorial objects via some specified operations. 
Especially, we consider reconfiguration problems for 
time-respecting arborescences, which were introduced by 
Kempe, Kleinberg, and Kumar. 
We first prove that 
if the roots of the initial and target time-respecting arborescences 
are the same, then the target arborescence is always reachable from 
the initial one and we can 
find a shortest reconfiguration sequence in polynomial time. 
Furthermore, we show 
if the roots 
are not the same, then the target arborescence may not be 
reachable from the initial one.
On the other hand, we show that we can determine whether the target 
arborescence is  
reachable form the initial one in polynomial time. 
Finally, we prove that 
it is NP-hard to find a shortest reconfiguration sequence
in the case where the roots 
are not the same. 
Our results show an interesting contrast to 
the previous results for 
(ordinary) arborescences
reconfiguration problems.
\end{abstract}

\section{Introduction}

An arborescence, which is a directed analogue of a spanning tree in 
an undirected graph, is one of the most fundamental combinatorial objects 
in a digraph.
For example, the problem of finding a minimum-cost arborescence 
in a digraph with a specified root vertex has been 
extensively studied 
(see, e.g., \cite{C65,E67,G03}). 
Furthermore, the theorem on arc-disjoint arborescences 
proved by Edmonds~\cite{E73} is one of the most important 
results in graph theory. 

Motivated by a variety of settings, such as communication in distributed networks, epidemiology and scheduled transportation networks, Kempe, Kleinberg and Kumar~\cite{KempeKK02} introduced the concept of temporal networks, which can be used to 
analyze relationships involving time over networks. 
More formally, 
a temporal network consists of a digraph $D = (V,A)$ and 
a time label function $\lambda \colon A \to \mathbb{R}_+$, where 
$\mathbb{R}_+$ denotes the set of non-negative real numbers. 
For each arc $a \in A$, 
the value $\lambda(a)$ 
specifies the time at which the two end vertices of 
$a$ can communicate.
Thus, in an arborescence $T$ with a root $r$, 
in order that 
the root $r$ can send information to every vertex, 
for every vertex $v$, 
the time labels of the arcs of the directed path from $r$ to $v$ must 
be non-decreasing. 
An arborescence $T$ in $D$ satisfying this property 
is called a \emph{time-respecting} arborescence. 
For example, \figurename~\ref{fig:example} illustrates four different time-respecting arborescences in a digraph.  

\begin{figure}
    \centering
    \includegraphics{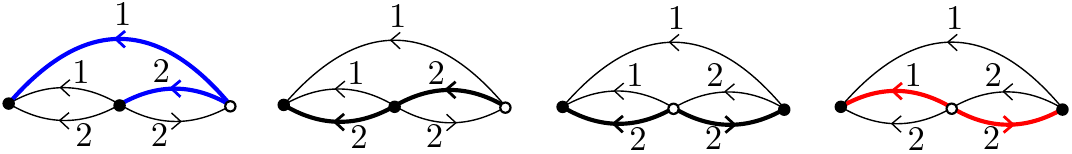}
    \caption{Examples of time-respecting arborescences.}
    \label{fig:example}
\end{figure}

In this paper, we study time-respecting arborescences in digraphs 
from the viewpoint of combinatorial 
reconfiguration.
Combinatorial reconfiguration~\cite{IDHPSUU11,N18} analyzes the reachability (and its related questions) of the solution space formed by combinatorial objects under a prescribed adjacency relation.
The algorithmic studies of combinatorial reconfiguration 
were initiated by Ito~et~al.~\cite{IDHPSUU11}, and
have been actively studied for this decade. 
(See, e.g., a survey~\cite{N18}.) 

\subsection{Our problem and related work}
In this paper, we introduce the \textsc{Time-Respecting Arborescence Reconfiguration} problem, as follows:
Given two time-respecting arborescences in a digraph $D$, we are asked to determine whether or not we can transform one into the other by exchanging a single arc in the current arborescence at a time, so that all intermediate results remain time-respecting arborescences in $D$. 
(We call this sequence of arborescences a \emph{reconfiguration sequence}.)
For example, \figurename~\ref{fig:example} shows such a transformation between the blue and red arborescences, and hence it is a yes-instance. 

This is the first paper, as far as we know, which deals with the \textsc{Time-Respecting Arborescence Reconfiguration} problem. 
However, reconfiguration problems have been studied for spanning trees and arborescences without time-respecting condition. 
For undirected graphs without time-respecting condition, 
it is well-known that every two spanning trees can be transformed into each other by exchanging a single edge at a time, because the set of spanning trees forms the family of bases of a matroid~\cite{Edmonds71,IDHPSUU11}.
For digraphs without time-respecting condition, Ito~et~al.~\cite{ItoIKNOW23} proved that every two arborescences can be transformed into each other by exchanging a single arc at a time.\footnote{Note that, in the paper~\cite{ItoIKNOW23}, an arborescence is not necessarily a spanning subgraph. Arborescence in our paper corresponds to spanning arborescence in~\cite{ItoIKNOW23}.} 
Interestingly, we note that this property does not hold when considering the time-respecting condition. (See \figurename~\ref{fig:no-instance}.)
Furthermore, in both undirected~\cite{IDHPSUU11} and directed~\cite{ItoIKNOW23} cases, shortest transformations can be found in polynomial time, because there is a transformation that exchanges only edges (or arcs) in the symmetric difference of two given spanning trees (resp.~arborescences).  

\begin{figure}
    \centering
    \includegraphics{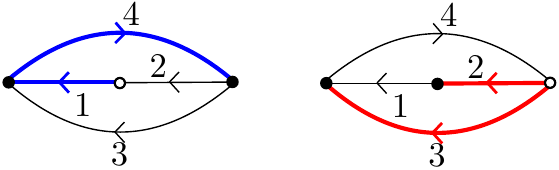}
    \caption{There is no desired transformation from the time-respecting arborescence induced by blue arcs to the time-respecting arborescence induced by red arcs.}
    \label{fig:no-instance}
\end{figure}

We here review some further previous work. 
Reconfiguration problems have been studied for various kinds of combinatorial objects, mainly in undirected graphs. 
Reconfiguration of spanning trees with additional constrains was studied in \cite{Bousquet20ESA,Bousquet22STACS}. 
Hanaka et~al.~\cite{HIMMNSSV20} introduced the framework of subgraph reconfiguration problems, and studied the problem for several combinatorial objects, including trees:
they showed that every two trees (that are not necessarily spanning) in an undirected graph can be transformed into each other by exchanging a single edge at a time unless two input trees have different numbers of edges.
Motivated by applications in motion planning, Biasi and Ophelders~\cite{BO18}, Demaine~et al.~\cite{DEHJLUU19}, and Gupta et~al.~\cite{GSZ20} studied some variants of reconfiguring undirected paths.
These variants are shown to be PSPACE-complete in general, while they are fixed-parameter tractable when parameterized by the length of input paths.

\subsection{Our contribution}
Our contributions are summarized as follows. 
We first prove that 
if the roots of the initial and target time-respecting arborescences 
are the same, then the target arborescence is always reachable from 
the initial one and we can 
find a shortest reconfiguration sequence in polynomial time. 
Furthermore, we show 
if the roots of the initial and target time-respecting arborescences 
are not the same, then the target arborescence may not be 
reachable from the initial one.
On the other hand, we show that we can determine whether the target 
arborescence is  
reachable from the initial one in polynomial time. 
Finally, we prove that 
it is NP-hard to find a shortest reconfiguration sequence
in the case where the roots of the initial and target time-respecting arborescences 
are not the same. 
Our results show an interesting contrast to 
the results for 
(ordinary) arborescence
reconfiguration problems~\cite{ItoIKNOW23}.  
See Table~\ref{table:summary} for the summary of our results. 

\begin{table}[t]
\centering 
\caption{
The column Reachability shows the answer for the question whether 
every arborescence is reachable from any other arborescence.
The column Reachability shows the results for the problem of determining 
whether the target arborescence is reachable from 
the initial one.
The symbol P means that the problem can be solved in polynomial time.
The column Shortest variant shows the results for the problem of 
finding a shortest reconfiguration sequence.}
\begin{tabular}{|r||l|l|}
\hline
& Reachability & Shortest sequence \\  
\hline \hline  
Arborescences without time-respecting & always yes~\cite{ItoIKNOW23} & P~\cite{ItoIKNOW23} \\
\hline 
Identical roots with time-respecting & always yes [Thm~\ref{thm:identical}] & P [Thm~\ref{thm:identical}]\\
\hline 
Non-identical roots with time-respecting & P [Thm~\ref{thm:non-identical}] & NP-complete [Thm~\ref{thm:shortest}] \\
\hline 
\end{tabular} 
\label{table:summary}
\end{table}

\section{Preliminaries}\label{sec:preli}
Let $D = (V, A)$ be a digraph with possibly multiple arcs.
We write $V(D)$ and $A(D)$ to denote the vertex set and arc set of $D$, respectively. 
For an arc $e = (u, v) \in A$, we call $v$ the \emph{head} of $e$, denoted $\head(e)$, and $u$ the \emph{tail} of $e$, denoted $\tail(e)$.
For $v \in V$, we denote by $\delta^-_D(v)$ the set of arcs incoming to $v$ in $D$ (i.e. $\delta^-_D(v) = \{e \mid e = (u, v) \in A\}$) and by $\delta^+_D(v)$ the set of outgoing arcs from $v$ (i.e. $\delta^+_D(v) = \{e \mid e = (v, u) \in A\}$).
We extend these notations to sets: For $X \subseteq V$, we define $\delta^-_D(X) = \{e = (u, v) \in E \mid u \in V \setminus X, v \in X\}$ and $\delta^+_D(X) = \{e = (u, v) \in E \mid u \in X, v \in V \setminus X\}$.
We may omit the subscript when no confusion arises.
For $e \in A$ (resp.~$f \in V\times V$), we denote by $D - e$ (resp.~$D + f$) the digraph obtained from $D$ by removing $e$ (resp.~adding $f$).

Let $r \in V$.
An \emph{r-arborescence} in $D$ is a spanning acyclic subgraph of $D$ in which there is exactly one (directed) path to any vertex from $r$.
An \emph{arborescence} (without specifying $r$) in $D$ is an $r$-arborescence for some $r \in V$.
Let $\lambda\colon A \to \mathbb R_{+}$.
For a directed path $P$ that traverses arcs $e_1, e_2, \dots , e_k$ in this order, 
we say that $P$ is \emph{time-respecting for $\lambda$} if $\lambda(e_{i}) \le \lambda(e_{i+1})$ for $1 \le i < k$.
An ($r$-)arborescence $T$ is time-respecting for $\lambda$ if every directed path in it is time-respecting.
When $\lambda$ is clear from the context, we may just say $P$ (or $T$) is time-respecting.

For two arborescences $T_1$ and $T_2$ in $D$, a \emph{reconfiguration sequence} between $T_1$ and $T_2$ is a sequence of arborescences $(T^0, T^1, \ldots, T^\ell)$ in $D$ with $T^0 = T_1$ and $T^\ell = T_2$ such that for $0 \le i < \ell$, $|A(T_i) \setminus A(T_{i+1})| = |A(T_{i+1}) \setminus A(T_{i})| = 1$.
In other words, $T_{i+1}$ is obtained from $T_i$ by removing an arc $e \in A(T_i)$ and adding an arc $f \notin A(T_i)$ (i.e., $T_{i+1} = T_i - e + f$).
The \emph{length} of the reconfiguration sequence is defined as $\ell$. 

\section{Minimal time-respecting \texorpdfstring{$r$-arborescences}{}}\label{sec:minimal}

In this section, we give a polynomial-time algorithm for computing a \emph{minimal} time-respecting $r$-arborescence of a digraph $D = (V, A)$ for $\lambda$.
This arborescence plays a vital role in the subsequent section.

Let $r \in V$.
We assume that $D$ has at least one time-respecting $r$-arborescence.
A function $d\colon V \to \mathbb R_+$ is defined as
\begin{align*}
    d(v) = \min\{\lambda(e) \mid e \in \delta^-(v) \text{ and } \exists \text{ time-respecting path from } r \text{ containing } e\},
\end{align*}
where we define $d(r) = 0$. 
Since every directed path from $r$ in the $r$-arborescence is time-respecting, this function $d$ is well-defined (under the assumption that $D$ has at least one time-respecting $r$-arborescence).
A time-respecting arborescence $T$ is said to be \emph{minimal} if the unique arc $e$ of $T$ incoming to $v$ satisfies $\lambda(e) = d(v)$ for $v \in V \setminus \{r\}$.
Now, we claim that (under the assumption) $D$ has a minimal time-respecting $r$-arborescence, which can be found by the following algorithm.

\begin{enumerate}
    \item Set $R := \{r\}$, $T := \emptyset$, and $d'(r) := 0$.
    \item Repeat the following two steps until $R = V$.
    \item Let $e \in \delta^+_D(R)$ minimizing $\lambda(e)$ subject to $\lambda(e) \ge d'(\tail(e))$. If there is no such an arc $e \in\delta^+_D(R)$, the algorithm halts.
    \item $R \coloneqq R \cup \{\head(e)\}$, $T \coloneqq T \cup \{e\}$, and $d'(\head(e)) \coloneqq \lambda(e)$.
\end{enumerate}

Note that $T$ is an arc set in the algorithm, which is sometimes identified with the subgraph induced by $T$. 
In the following, we use $R$, $T$, and $d'$ to denote the values of $R$, $T$, and $d'$ after the execution of the above algorithm, respectively.

\begin{lemma}\label{lem:minimal}
    Suppose that $D$ has a time-respecting $r$-arborescence.
    Then, $T$ forms a minimal time-respecting $r$-arborescence of $D$.
\end{lemma}

\begin{proof}
    By the construction of $T$, every (directed) path from $r$ in $T$ is time-respecting, and hence we have $d'(v) \ge d(v)$ for $v \in R$.
    For $1 \le i \le |T|$, we let $e_i$ be the arc selected at Step 3 in the $i$th iteration of the execution.
    Let $v_0 = r$ and let $v_i = \head(e_i)$.
    In the following, we first show, by induction, that $d'(v_i) = d(v_i)$ for all $0 \le i \le |T|$.
    The base case $i = 0$ is clear from the definition (i.e., $d'(r) = d(r)$).

    Let $i \ge 1$ and let $R_i = \{v_j \mid 0 \le j < i\}$ be the set of vertices that are ``reached'' from $r$ before the $i$th iteration.
    Suppose for contradiction that $d'(v_i) > d(v_i)$.
    Let $P$ be a time-respecting path from $r$ to $v_i$ in which the unique arc $e$ incoming to $v_i$ satisfies $\lambda(e) = d(v_i)$.
    Since $r \in R_i$ and $v_i \notin R_i$, there is an arc $e'$ on $P$ such that $\tail(e') \in R_i$ and $\head(e') \notin R_i$, i.e., $e' \in \delta^+_D(R_i)$.
    For such an arc $e'$, we have $\lambda(e') \leq \lambda(e) = d(v_i) < d'(v_i) = \lambda(e_i)$.
    This implies that $e_i$ cannot be selected at Step 3 in the $i$th iteration, 
    because $e' \in \delta^+(R_i)$ and $\lambda(e') \ge d(\tail(e')) = d'(\tail(e'))$ by the induction hypothesis. 

    We next show that $R = V$, which implies that $T$ forms a minimal time-respecting $r$-arborescence of $D$.
    As we have shown above, $d'(v) = d(v)$ for all $v \in R$.
    Suppose for contradiction that $R \neq V$.
    We note that every arc $e \in \delta_D^+(R)$ satisfies that $\lambda(e) < d'(\tail(e))$.
    Let $v \in V \setminus R$ and let $P$ be an arbitrary time-respecting path from $r$ to $v$ in $D$.
    We can choose such a vertex $v$ in such a way that all vertices of $P$ except for $v$ belong to $R$, as $r \in R$ and any subpath of $P$ starting from $r$ is also time-respecting.
    Let $e$ be the unique arc of $P$ incoming to $v$.
    Since $P$ is time-respecting and $\tail(e) \in R$, $\lambda(e) \ge d(\tail(e)) = d'(\tail(e))$.
    This contradicts the fact that every arc $e' \in \delta_D^+(R)$ satisfies $\lambda(e') < d'(\tail(e'))$.
    Therefore, the lemma follows.
\end{proof}

As a consequence of this lemma, we have the following corollary.

\begin{corollary}\label{cor:minimal}
  In polynomial time, we can compute a minimal time-respecting $r$-arborescence of $D$ or 
  conclude that $D$ has no time-respecting $r$-arborescence.
\end{corollary}

\section{Time-respecting \texorpdfstring{$r$-arborescence}{} reconfiguration}\label{sec:fixed-root-case}

In this section, we give a polynomial-time algorithm for finding a shortest reconfiguration sequence between given two time-respecting $r$-arborescences of a graph $D$ such that all intermediates are also time-respecting $r$-arborescences of $D$ (if any).
We show, in fact, that such a (shortest) reconfiguration sequence between $T_1$ and $T_2$ always exists, in contrast to the fact that, for some digraph, there is no reconfiguration sequence between time-respecting arborescences of distinct roots (see \figurename~\ref{fig:no-instance}).  

Let $D = (V, A)$ be a digraph with $\lambda\colon A \to \mathbb R_+$ and let $r \in V$.
Let $T_1$ and $T_2$ be time-respecting $r$-arborescences of $D$.
We construct a digraph $D^* = (V, A^*)$, where $A^* = A(T_1) \cup A(T_2)$.
As there is a time-respecting $r$-arborescence of $D^*$ (say, $T_1$), a minimal time-respecting $r$-arborescence $T^*$ of $D^*$ can be computed in polynomial time with the algorithm in Corollary~\ref{cor:minimal}.

To show the existence of a reconfiguration sequence between $T_1$ and $T_2$ in $D$, it suffices to give a reconfiguration sequence between $T_1$ and $T^*$ in $D^*$ as $T_2$ is symmetric and $D^*$ is a subgraph of $D$.

We transform $T_1$ into the minimal time-respecting $r$-arborescence $T^*$ as follows.
Let $k = |A(T^*)|$.
Let $e_1, e_2, \ldots, e_{k}$ be the arcs of $T^*$ such that $e_i$ is selected at Step 3 in the $i$th iteration of the algorithm in Corollary~\ref{cor:minimal}.
We set $T^0 = T_1$.
For $1 \le i \le k$ in increasing order, we define $T^i = T^{i-1} - f_i + e_i$ (possibly $e_i=f_i$), where $f_i$ is the unique arc of $T_1$ incoming to $\head(e_i)$.
Clearly, $T^k = T^*$.
Since the update operation preserves the indegree of each vertex
and the reachability of each vertex from $r$,
every $T^{i}$ is always an $r$-arborescence of $D^*$.
Moreover, the following lemma ensures that $T^{i}$ is time-respecting.

\begin{lemma}\label{lem:r-Arb:time-respecting}
    For $0 \le i \le k$, $T^i$ is a time-respecting $r$-arborescence of $D^*$.
\end{lemma}
\begin{proof}
    It suffices to show that for $1 \le i \le k$, $T^i$ is time-respecting, assuming that $T^{i-1}$ is time-respecting.
    It suffices to consider the case when $f_i \neq e_i$. 
    Since $\head(e_i) = \head(f_i)$, $T^{i-1} + e_i$ has exactly two paths $P$ and $P'$ from $r$ to $\head(e_i)$, where $P$ (resp.~$P'$) is the one that contains $e_i$ (resp.~$f_i$).
    Since $T^{i-1}$ is time-respecting, $P'$ is indeed time-respecting.
    Similarly, since $P (\subseteq \{ e_1, e_2, \dots, e_i \})$ is a path in $T^*$, it is also time-respecting.
    As $T^*$ is minimal, we have $d(\head(f_i)) = \lambda(e_i)$, which implies that $\lambda(f_i) \ge \lambda(e_i)$.
    Therefore, $T^{i}$ is time-respecting.
\end{proof}

As $T^k = T^*$, this lemma shows that there is a reconfiguration sequence between $T_1$ and $T^*$. 
Since we update $T \leftarrow T - e_i + f_i$ only when $e_i \neq f_i$, 
the obtained sequence has length $|A(T^*) \setminus A(T_1)|$. 
Similarly, there is a sequence between $T_2$ and $T^*$ of length $|A(T^*) \setminus A(T_2)|$. 
By combining them, we obtain a reconfiguration sequence from $T_1$ to $T_2$ of length $|A(T^*) \setminus A(T_1)| + |A(T^*) \setminus A(T_2)|$. 

We now show that this length is equal to $|A(T_1) \setminus A(T_2)|$, which implies that the sequence is shortest among all reconfiguration sequences between $T_1$ and $T_2$.
For any $e \in A(T_1) \cap A(T_2)$, since $e$ is a unique arc entering $\head(e)$ in $D^*$, we obtain $e \in A(T^*)$. This means that $A(T_1) \cap A(T_2) \subseteq A(T^*)$. 
Then, we obtain 
\begin{align*}
& |A(T^*) \setminus A(T_1)| + |A(T^*) \setminus A(T_2)| 
= |A(T^*) \setminus (A(T_1) \cap A(T_2))| \\
&= |A(T^*)| - |A(T_1) \cap A(T_2)| 
= |A(T_1) \setminus A(T_2)|, 
\end{align*}
where we use $A(T^*) \subseteq A(T_1) \cup A(T_2)$ in the first equality, 
use $A(T_1) \cap A(T_2) \subseteq A(T^*)$ in the second equality, and 
use $|A(T^*)| = |A(T_1)|$ in the last equality. 

\begin{theorem}\label{thm:identical}
    There is a reconfiguration sequence between two time-respecting $r$-arborescence $T_1$ and $T_2$ of $D$ with length $|A(T_1) \setminus A(T_2)|$.
    Moreover, such a reconfiguration sequence can be found in polynomial time.
\end{theorem}

\section{Time-respecting arborescence reconfiguration}\label{sec:unfixed-root-case}
For two time-respecting arborescences $T_1$ and $T_2$ in a digraph $D=(V,A)$ with $\lambda: A\rightarrow \mathbb{R}_+$,
we consider the problem of determining whether there exists a reconfiguration sequence from $T_1$ to $T_2$,
which we call \textsc{Time-respecting Arborescence Reconfiguration}.
By Theorem~\ref{thm:identical}, for any two $r$-arborescences $T_1$ and $T_2$, there exists a reconfiguration sequence from $T_1$ to $T_2$, that is,
the answer of \textsc{Time-respecting $r$-Arborescence Reconfiguration} is always yes.
However, \textsc{Time-respecting Arborescence Reconfiguration} does not have the property, that is,
there is a no-instance in the problem (see Fig.~\ref{fig:no-instance}).
In this section, we show that the problem can be solved in polynomial time.
\begin{theorem} \label{thm:non-identical}
We can solve \textsc{Time-respecting Arborescence Reconfiguration} in polynomial time.
\end{theorem}

To prove the theorem, we introduce some concepts.  
For $t \in \mathbb{R}_+$, we say that a subgraph $H=(V(H), A(H))$ of $D$ is \emph{$t$-labeled extendible} if
    \begin{enumerate}
        \item $\lambda(e) = t$ for all $e \in A(H)$ and 
        \item the digraph $D^{\prime}$ obtained from $D$ by contracting $H$ into a vertex $r_H$ has a time-respecting $r_H$-arborescence $T$ 
              such that $\lambda(e) \geq t$ for $e \in A(T)$.
    \end{enumerate}
We can see that if a $t$-labeled extendible subgraph $H$ contains an arborescence $T_H$, then $A(T_H) \cup A(T)$ induces a time-respecting arborescence in $D$. 

Let $D=(V,A)$ be a digraph and
define the \emph{reconfiguration graph} $\mathcal{G}(D)$ as follows: the vertex set consists of all the time-respecting arborescences of $D$ and two time-respecting arborescences are joined by an (undirected) edge in the reconfiguration graph if and only if one is obtained from the other by exchanging a single arc.
For a vertex $r \in V$, let $\mathcal{G}_r$ be the subgraph of $\mathcal{G}(D)$ induced by the time-respecting $r$-arborescences. By Theorem~\ref{thm:identical}, $\mathcal{G}_r$ is connected for $r \in V(D)$.
Let $\mathcal{G}^{\prime}(D)$ be the graph obtained from $\mathcal{G}(D)$ by contracting $\mathcal{G}_r$ into a vertex $v_r$ for each $r \in V(D)$.
We show the following necessary and sufficient condition for two vertices in $\mathcal{G}^{\prime}(D)$ to be adjacent to each other.

\begin{lemma}\label{thm:r_arborescence}
Let $r_1$ and $r_2$ be two distinct vertices in a digraph $D=(V,A)$ with $\lambda: A\rightarrow \mathbb{R}_+$.
Then $v_{r_1}$ and $v_{r_2}$ are adjacent in $\mathcal{G}^{\prime}(D)$ if and only if
one of the following holds:
\begin{enumerate}
    \item[\upshape(i)] there exist an arc $f=(r_2,r_1)$ and a time-respecting $r_1$-arborescence $T_1$ such that
    $\lambda(f) \leq \lambda(e^{\prime}) $ for each $e^{\prime} \in \delta^+_{T_1}(r_1)\setminus \delta^-_{T_1}(r_2)$,
    \item[\upshape(ii)] there exist an arc $e=(r_1,r_2)$ and a time-respecting $r_2$-arborescence $T_2$ such that
    $\lambda(e) \leq \lambda(e^{\prime}) $ for each $e^{\prime} \in \delta^+_{T_2}(r_2)\setminus \delta^-_{T_2}(r_1)$,
    \item[\upshape(iii)] 
    for some $t \in \mathbb{R}_+$, $D$ has a $t$-labeled extendible directed cycle $C$ that contains both $r_1$ and $r_2$. 
\end{enumerate}
\end{lemma}

\begin{proof}
\medskip
\noindent 
\textbf{[Necessity (``only if'' part)]} Suppose that $v_{r_1}$ and $v_{r_2}$ are adjacent in $\mathcal{G}^{\prime}(D)$.
Then there exist a time-respecting $r_1$-arborescence $T_1$ and two arcs $e$ and $f$ in $A(D)$ such that
$T_1-e+f$ is a time-respecting $r_2$-arborescence. Let $T_2:=T_1-e+f$. Since $T_1$ is an $r_1$-arborescence and $T_2$ is an $r_2$-arborescence,
$\head(e)=r_2$ and $\head(f)=r_1$.
Then, $T_1+f$ contains a directed cycle $C$, which contains $e$. 
(Otherwise, $T_2$ contains a directed cycle, which contradicts the fact that $T_2$ is an arborescence.)
Let $\ell$ denote the length of $C$. 
Suppose that $C$ traverses arcs $e_1, e_2, \dots , e_\ell$ in this order when starting from $r_1$, that is, 
$e_\ell=f$ and $\head(e_i)=\tail(e_{i+1})$ for each $i$, where $e_{\ell+1}=e_1$.
We can easily see that (i) holds if $f=(r_2,r_1)$ and (ii) holds if $e=(r_1,r_2)$.

Hence it suffices to consider the case when $f \neq (r_2,r_1)$ and $e \neq (r_1,r_2)$.
Let $j$ be the index such that $e_j = e$. 
Note that $j \neq 1$ by $e \neq (r_1,r_2)$, $j \neq \ell-1$ by $f \neq (r_2,r_1)$, and $j \neq \ell$ by $e \neq f$. 
Since $T_1$ is a time-respecting $r_1$-arborescence, we obtain the following inequalities.
\begin{align}
   \lambda(e_1) \leq \lambda(e_2)\leq \cdots \leq\lambda(e_{j+1})\leq \cdots \leq \lambda(e_{\ell-1}) \label{eq1}
\end{align}
Since $T_2$ is a time-respecting $r_2$-arborescence, we obtain the following inequalities.
\begin{align}
   \lambda(e_{j+1}) \leq \lambda(e_{j+2})\leq \cdots \leq\lambda(e_\ell)\leq \lambda(e_1) \label{eq2} 
\end{align}
By (\ref{eq1}) and (\ref{eq2}), we obtain $\lambda(e_1)=\lambda(e_2)=\cdots=\lambda(e_\ell)$.
Let $D^{\prime}$ be the digraph obtained from $D$ by contracting $C$ into a vertex $r_C$.
Then the digraph induced by the arcs in $A(T_1) \setminus A(C)$ is a time-respecting $r_C$-arborescence in $D^{\prime}$.
Hence (iii) holds.

\medskip
\noindent 
\textbf{[Sufficiency (``if'' part)]}
Suppose that (i) holds. Let $e$ be the arc in $T_1$ such that $\head(e)=r_2$.
Then $T_1-e+f$ is a time-respecting $r_2$-arborescence in $D$.
If (ii) holds, then we obtain a time-respecting $r_1$-arborescence in $D$ from $T_2$ by one step by a similar argument above.

Suppose that (iii) holds. 
Let $e$ and $f$ be the arcs in $C$ such that $\head(e)=r_2$ and $\head(f)=r_1$, respectively. 
Let $D^{\prime}$ be the digraph obtained from $D$ by contracting $C$ into $r_C$, and 
let $T_C$ be a time-respecting $r_C$-arborescence in $D^{\prime}$.
Then the digraph $T_1$ induced by the arcs in $A(T_C) \cup A(C)\setminus \{f\}$ is a time-respecting $r_1$-arborescence in $D$ and
the digraph $T_2$ induced by the arcs in $A(T) \cup A(C)\setminus \{e\}$ is a time-respecting $r_2$-arborescence in $D$.
Since $T_2=T_1-e+f$, $v_{r_1}$ and $v_{r_2}$ are adjacent in $\mathcal{G}^{\prime}(D)$.

This completes the proof of Lemma~\ref{thm:r_arborescence}.
\end{proof}

Let $r_1$ and $r_2$ be two distinct vertices in $D$. 
To check condition (i) in Lemma~\ref{thm:r_arborescence}, it suffices to give a polynomial-time algorithm for finding a time-respecting $r_1$-arborescence $T_1$ such that $\lambda(e') \ge \lambda(f)$ for each $e' \in \delta^+_{T_1}(r_1) \setminus \delta^-_{T_1}(r_2)$, where $f = (r_2, r_1)$.
This can be done in polynomial time
by removing all the arcs $e \in \delta^+_{D}(r_1) \setminus \delta^-_D(r_2)$ with $\lambda(e) < \lambda(f)$ from $D$ and 
by applying Corollary~\ref{cor:minimal} to find 
a time-respecting $r_1$-arborescence in the obtained digraph. 
Similarly, condition (ii) in Lemma~\ref{thm:r_arborescence} can be checked in polynomial time.

We consider to check condition (iii) in Lemma~\ref{thm:r_arborescence}.
However, it is NP-hard to find a directed cycle in a digraph containing two specified vertices~\cite{FORTUNE1980111}. 
To overcome this difficulty, we consider a supergraph of $\mathcal{G}^{\prime}(D)$, 
which is a key ingredient in our algorithm. 

For $t \in \mathbb{R}_+$, let $D_t$ denote the subgraph of $D$ induced by the edges of label $t$. 
We consider the following condition instead of (iii): 
\begin{enumerate}
    \item[\upshape(iii)'] 
    for some $t \in \mathbb{R}_+$, 
    $r_1$ and $r_2$ are contained in the same strongly connected component in $D_t$, which is $t$-labeled extendible. 
\end{enumerate}
We can see that (iii)' is a relaxation of (iii) as follows. 
If $r_1$ and $r_2$ are contained in a $t$-labeled extendible directed cycle $C$, then 
they are contained in the same connected component $H$ of $D_t$. 
Furthermore, since $C$ is $t$-labeled extendible, so is $H$, which means that (iii)' holds. 

Define $\hat{\mathcal{G}}(D)$ as the graph whose vertex set is the same as $\mathcal{G}^{\prime}(D)$, and 
$v_{r_1}$ and $v_{r_2}$ are adjacent in $\mathcal{G}^{\prime}(D)$ 
if and only if $r_1$ and $r_2$ satisfy (i), (ii), or (iii)'. 
Since (iii)' is a relaxation of (iii), Lemma~\ref{thm:r_arborescence} shows that 
$\hat{\mathcal{G}}(D)$ is a supergraph of $\mathcal{G}^{\prime}(D)$. 
We now show the following lemma. 

\begin{lemma}\label{relax01}
    Let $r_1$ and $r_2$ be distinct vertices in $D$ that satisfy condition (iii)'. 
    Then, $\mathcal{G}^{\prime}(D)$ contains a path between $v_{r_1}$ and $v_{r_2}$. 
\end{lemma}

\begin{proof}
    Let $H$ be the $t$-labeled extendible strongly connected component in $D_t$ that contains $r_1$ and $r_2$, where $t \in \mathbb{R}_+$.
    Since $H$ is strongly connected, it contains a directed path from $r_1$ to $r_2$ that traverses vertices $p_0, p_1, \dots , p_k$ in this order, where $p_0 = r_1$ and $p_k = r_2$. Then, for $0 \le i \le k-1$, $H$ has a directed cycle $C_i$ containing arc $(p_i, p_{i+1})$ as $H$ is strongly connected.
    Since $H$ is $t$-labeled extendible and strongly connected, we see that $C_i$ is also $t$-labeled extendible. 
    Therefore, condition (iii) in Lemma~\ref{thm:r_arborescence} shows that 
    $v_{p_i}$ and $v_{p_{i+1}}$ are adjacent in $\mathcal{G}^{\prime}(D)$, 
    which implies that $\mathcal{G}^{\prime}(D)$ contains a path connecting $v_{r_1}=v_{p_0}$ and $v_{r_2}=v_{p_k}$.
\end{proof}

By this lemma and by the fact that $\hat{\mathcal{G}}(D)$ is a supergraph of $\mathcal{G}^{\prime}(D)$, 
we obtain the following lemma. 

\begin{lemma}\label{relax02}
    For any distinct vertices $r_1$ and $r_2$ in $D$, 
    $\hat{\mathcal{G}}(D)$ has a $v_{r_1}$-$v_{r_2}$ path if and only if $\mathcal{G}^{\prime}(D)$ has one. 
\end{lemma}

This lemma shows that it suffices to check the reachability in $\hat{\mathcal{G}}(D)$
to solve \textsc{Time-respecting Arborescence Reconfiguration}. 
For distinct vertices $r_1$ and $r_2$ in $D$, (i) and (ii) can be checked in polynomial time as described above. 
We can check condition (iii)' by applying the following algorithm 
for each $t \in \{\lambda(e) \mid e \in A(D)\}$. 

\begin{enumerate}
    \item Construct the subgraph $D_t$ of $D$ induced by the edges of label $t$.
    \item If $r_1$ and $r_2$ are contained in the same strongly connected component $H$ in $D_t$, then go to Step 3. 
    Otherwise, (iii)' does not hold for the current $t$. 
    \item Contract $H$ into a vertex $r_H$ to obtain a digraph $D'$. 
    Remove all the arcs $e \in A(D')$ with $\lambda(e) < t$ from $D'$ and find a time-respecting $r_H$-arborescence in this digraph by Corollary~\ref{cor:minimal}. 
    \item If a time-respecting $r_H$-arborescence is found, then $r_1$ and $r_2$ satisfy condition (iii)'. 
    Otherwise, (iii)' does not hold for the current $t$.
\end{enumerate}

Since we can decompose a digraph into strongly connected components in polynomial time, 
this algorithm runs in polynomial time. 
Therefore, by checking (i), (ii), and (iii)' for every pair of $r_1$ and $r_2$, 
we can construct $\hat{\mathcal{G}}(D)$ in polynomial time. 

By Theorem~\ref{thm:identical} and Lemma~\ref{relax02}, 
a time-respecting $r_1$-arborescence $T_1$ can be reconfigured to 
a time-respecting $r_2$-arborescence $T_2$ 
if and only if $\hat{\mathcal{G}}(D)$ contains a $v_{r_1}$-$v_{r_2}$ path, 
which can be checked in polynomial time. 
This completes the proof of Theorem~\ref{thm:non-identical}. 

\section{NP-completeness of shortest reconfiguration}

For two time-respecting arborescences $T_1$ and $T_2$ in a digraph $D=(V, A)$ with $\lambda\colon A \to \mathbb R_{+}$ and for a positive integer $\ell$, we consider the problem of determining whether 
there exists a reconfiguration sequence from $T_1$ to $T_2$ of length at most $\ell$, 
which we call \textsc{Time-respecting Arborescence Shortest Reconfiguration}.  
Note that the length is defined as the number of swap operations, which is equal to the number of time-respecting arborescences in the sequence minus one. 
In this section, we prove the NP-completeness of this problem. 

\begin{theorem} \label{thm:shortest}
\label{thm:shortesthardness}
\textsc{Time-respecting Arborescence Shortest Reconfiguration}  is NP-complete.
\end{theorem}

\begin{proof}
The proof of Theorem~\ref{thm:non-identical} shows that  
if $T_1$ is reconfigurable to $T_2$, then there exists a reconfiguration sequence whose length is bounded by a polynomial in $|V|$. 
This implies that \textsc{Time-respecting Arborescence Shortest Reconfiguration} is in NP. 

To show the NP-hardness, we reduce \textsc{Vertex Cover} to \textsc{Time-respecting Arborescence Shortest Reconfiguration}. 
Recall that, in \textsc{Vertex Cover}, we are given a graph $G=(V, E)$ and a positive integer $k$, 
and the task is to determine whether $G$ contains a vertex cover of size at most $k$ or not. 

Suppose that $G=(V, E)$ and $k$ form an instance of \textsc{Vertex Cover}.
We construct a digraph $D=(W, A)$ with multiple arcs as follows: 
\begin{align*}
W &= \{r_1, r_2\} \cup \{w_v \mid v \in V\} \cup \{w_e \mid e \in E\}, \\
A_1 &= \{(r_1, w_e), (r_2, w_e) \mid e \in E\}, \ \ 
A_2 =\{(r_1, r_2), (r_2, r_1)\}, \\
A_3 &= \{a_v = (r_2, w_v) \mid v \in V\}, \ \
A_4 = \{ (w_v, w_e) \mid e \in \delta_G(v) \}, \\
A_5 &= \{a'_v = (r_2, w_v) \mid v \in V\}, \ \
A = A_1 \cup A_2 \cup A_3 \cup A_4 \cup A_5. 
\end{align*}
Here, $w_v$ and $w_e$ are newly introduced vertices associated with $v \in V$ and $e \in E$, respectively. 
Note that $a_v \in A_3$ and $a'_v \in A_5$ are distinct, that is, they form multiple arcs. 
For $i \in \{1, 2, 3, 4, 5\}$ and for $a \in A_i$, we define $\lambda(a) = i$. 
Let $\ell = 2|E| + 2k + 1$. 
Let
\begin{align*}
T_1 &= \{(r_1, r_2)\} \cup \{(r_1, w_e) \mid e \in E\} \cup A_5, \\ 
T_2 &= \{(r_2, r_1) \} \cup \{(r_2, w_e) \mid e \in E\} \cup A_5. 
\end{align*}
See~\Cref{fig:reduction} for an illustration.
One can easily see that $T_1$ and $T_2$ are time-respecting arborescences in $D$. 
This completes the construction of an instance of \textsc{Time-respecting Arborescence Shortest Reconfiguration}. 

\begin{figure}[t]
    \centering
    \includegraphics[width=0.5\textwidth]{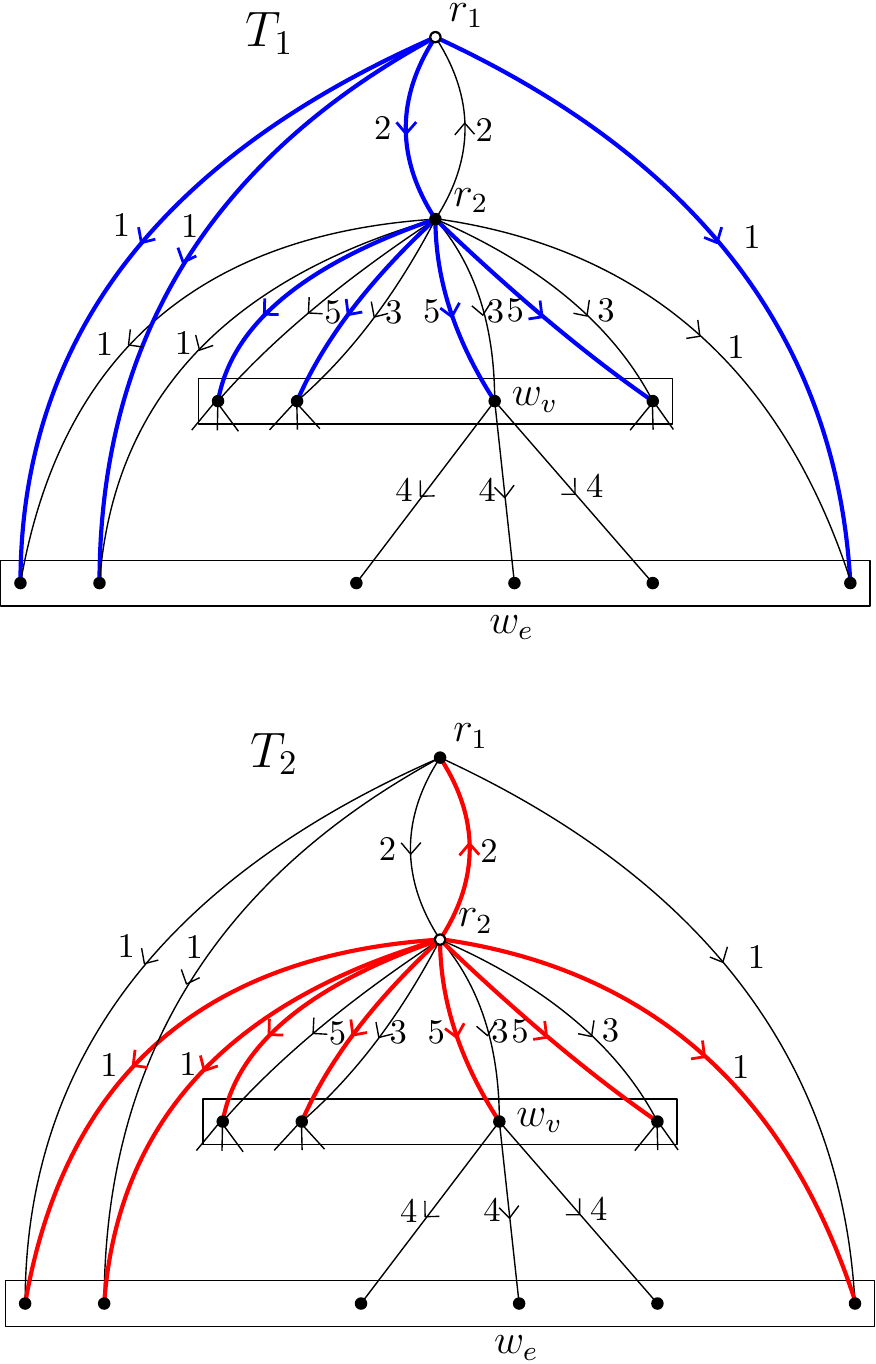}
    \caption{An illustration of the digraph $D$ constructed from $G$ and two time-respecting arborescences $T_1$ and $T_2$ in the proof of Theorem~\ref{thm:shortesthardness}.}
    \label{fig:reduction}
\end{figure}

To prove the validity of the reduction, it suffices to show that $G$ contains a vertex cover of size at most $k$ if and only if $D$ admits a reconfiguration sequence from $T_1$ to $T_2$ of length at most $\ell$. 

\medskip
\noindent 
\textbf{[Necessity (``only if'' part)]}
Suppose that $G$ contains a vertex cover $X \subseteq V$ with $|X| \le k$. 
For $e \in E$, let $\sigma(e)$ be an end vertex of $e$ that is contained in $X$. 
If both of the end vertices of $e$ are contained in $X$, then we choose one arbitrarily. 
Let
\begin{align*}
T'_1 &= \{(r_1, r_2)\} \cup \{a_v \mid v \in X\} \cup \{a'_v \mid v \in V \setminus X\} \cup \{(w_{\sigma(e)}, w_e) \mid e \in E \}, \\ 
T'_2 &= \{(r_2, r_1)\} \cup \{a_v \mid v \in X\} \cup \{a'_v \mid v \in V \setminus X\} \cup \{(w_{\sigma(e)}, w_e) \mid e \in E \}, 
\end{align*}
which are time-respecting arborescences in $D$. 
Then, $T_1$ can be transformed into $T'_1$ as follows: 
replace $a'_v$ with $a_v$ for each $v \in X$ and then
replace $(r_1, w_e)$ with $(w_{\sigma(e)}, w_e)$ for each $e \in E$. 
This shows that there exists a reconfiguration sequence from $T_1$ to $T'_1$ of length $|X| + |E|$. 
Similarly, $T'_2$ can be transformed into $T_2$ in $|X| + |E|$ steps. 
Since $T'_1$ and $T'_2$ are adjacent, by combining the above transformations, 
we obtain a reconfiguration sequence from $T_1$ to $T_2$ of length $2(|X| + |E|) +1 \le \ell$, 
which shows the necessity.

\medskip
\noindent 
\textbf{[Sufficiency (``if'' part)]}
Suppose that there is a reconfiguration sequence from $T_1$ to $T_2$ of length at most $\ell = 2 |E| + 2k +1$. 
Observe that, for any arborescence in $D$, its root is either $r_1$ or $r_2$, because $D$ contains no arc from $W \setminus \{r_1, r_2\}$ to $\{r_1, r_2\}$.  
This shows that 
the root has to move from $r_1$ to $r_2$ in the reconfiguration sequence, that is, 
the reconfiguration sequence contains two consecutive time-respecting arborescences $T'_1$ and $T'_2$ such that 
the root of $T'_i$ is $r_i$ for $i=1, 2$. 
Since $(r_1, r_2)$ is the unique arc entering $r_2$, $T'_1$ contains $(r_1, r_2)$. Similarly, $T'_2$ contains $(r_2, r_1)$. 
Therefore, $T'_2$ is obtained from $T'_1$ by removing $(r_1, r_2)$ and adding $(r_2, r_1)$, which means that 
$T'_1 = F \cup \{(r_1, r_2)\}$ and $T'_2 = F \cup \{(r_2, r_1)\}$ for some $F \subseteq A$. 
Since $T'_1$ and $T'_2$ are time-respecting and $\lambda((r_1, r_2)) = \lambda((r_2, r_1)) = 2$, we obtain $F \cap A_1 = \emptyset$. 
Hence, for each $e \in E$, $F$ contains a unique arc in $A_4$ entering $w_e$, 
that is, there exists $\sigma(e) \in V$ such that $(w_{\sigma(e)}, w_e) \in F$. 

Let $X = \{v \in V \mid a_v \in F\}$. We now show that $X$ is a vertex cover of size at most $k$. 
For each $e \in E$, since $\lambda((w_{\sigma(e)}, w_e)) < \lambda(a'_{\sigma(e)})$, $F$ does not contain $a'_{\sigma(e)}$. 
This shows $a_{\sigma(e)} \in F$, because $F$ contains exactly one arc entering $w_{\sigma(e)}$. 
Therefore, $\sigma(e) \in X$ for any $e \in E$, which implies that $X$ is a vertex cover in $G$. 

The length of the reconfiguration sequence from $T_1$ to $T'_1$ is at least $|T'_1 \setminus T_1| = |\{(w_{\sigma(e)}, w_e) \mid e \in E\} \cup  \{a_v \mid v \in X\}| = |E| + |X|$. 
Similarly, the length from $T'_2$ to $T_2$ is at least $|E| + |X|$. 
By considering a step from $T'_1$ to $T'_2$, the total length of the reconfiguration sequence from $T_1$ to $T_2$ is at least $2|E|+2|X|+1$. 
Since this length is at most $\ell$, we obtain $|X| \le k$. This shows that $G$ has a vertex cover of size at most $k$. 

\medskip

By the above argument, \textsc{Vertex Cover} is reduced to \textsc{Time-respecting Arborescence Shortest Reconfiguration}, 
and hence \textsc{Time-respecting Arborescence Shortest Reconfiguration} is NP-hard.
\end{proof}

\begin{remark}
\label{rem:numberoflabels}
The above proof works even if we perturb the value of $\lambda$, which shows that 
\textsc{Time-respecting Arborescence Shortest Reconfiguration} is NP-complete even when $\lambda(a) \neq \lambda(a')$ for distinct $a, a' \in A$. 
We also see that the above proof works even if we define 
$\lambda(a) = 1$ for $a \in A_1$, 
$\lambda(a) = 2$ for $a \in A_2 \cup A_3 \cup A_4$, and 
$\lambda(a) = 3$ for $a \in A_5$. 
This shows that 
\textsc{Time-respecting Arborescence Shortest Reconfiguration} is NP-complete even when $\lambda(a) \in \{1, 2, 3\}$ for $a \in A$. 
\end{remark}

\section{Concluding remarks}
As described in the introduction,  
if we remove the ``time-respecting'' constraint, then any arborescence $T_1$ can be transformed into another arborescence $T_2$ in $|T_1 \setminus T_2|$ steps; see~\cite{ItoIKNOW23}. 
This implies that \textsc{Time-respecting Arborescence Shortest Reconfiguration} 
can be solved in polynomial-time if $\lambda(a)$ takes the same value for any $a \in A$.
On the other hand, as described in Remark~\ref{rem:numberoflabels}, \textsc{Time-respecting Arborescence Shortest Reconfiguration} is NP-complete
even if $\lambda(a)$ takes one of the three given values for $a \in A$. 
It is open whether the shortest reconfiguration sequence can be found in polynomial time when $|\{\lambda(a) \mid a \in A\}| = 2$. 

\printbibliography

\end{document}